\documentclass{article}
\pdfoutput=1
\usepackage{amsfonts,pstricks}
\usepackage{amsmath, amsthm, amssymb,  mathrsfs, epsfig}
\usepackage{hyperref}
\usepackage{graphicx}
\usepackage{tikz}
\usepackage{caption, subcaption}
\usepackage{authblk}
\usepackage{bm}
\usepackage{bbm}
\usepackage{tensor}
\usetikzlibrary{decorations.markings,intersections}
\usepackage{enumitem}
\hypersetup{linktocpage = true}
\usepackage{tikz-cd}

\theoremstyle{definition}
\newtheorem{theorem}{Theorem}[section]

\newtheorem{conjecture}[theorem]{Conjecture}
\newtheorem{lemma}[theorem]{Lemma}

\newtheorem{proposition}[theorem]{Proposition}
\newtheorem{remark}[theorem]{Remark}
\newtheorem{definition}[theorem]{Definition}

\DeclareMathOperator{\Hom}{Hom}

\DeclareMathOperator{\diag}{diag}

\def\Real{\mathbb{R}}

\def\Complex{\mathbb{C}}

\def\CatB{\mathcal{B}}
\def\CatC{\mathcal{C}}

\def\U{\mathbf{U}}
\def\SU{\mathbf{SU}}

\def\unit{\mathbf{1}}

\def\ket#1{|#1\rangle}

\tikzset{
    partial ellipse/.style args={#1:#2:#3}{
        insert path={+ (#1:#3) arc (#1:#2:#3)}
    }
}

\newcommand{\Rttt}{R^{\tau \tau}_\tau}
\newcommand{\Rttone}{R^{\tau \tau}_{\mathbf{1}}}
\newcommand{\Sw}{\text{SWAP}}
\newcommand{\Ent}{{entangling leakage-free}}
\newcommand{\res}[1]{|_{#1}}

\begin{document}

\begin{titlepage}

\ \\

\begin{center}

{\LARGE The Search For Leakage-free Entangling Fibonacci Braiding Gates}

\vspace{0.5cm}
Shawn X. Cui,$^{1,2}$ Kevin T. Tian,$^3$ Jennifer F. Vasquez,$^4$ Zhenghan Wang,$^{3,5}$ Helen M. Wong$\,^6$

\vspace{5mm}

{\small
\textit{
$^1$Stanford Institute for Theoretical Physics, Stanford University, Stanford, CA 94305}\\

\vspace{2mm}

\textit{
$^2$Department of Mathematics, Virginia Tech, Blacksburg, VA, 24061 }\\

\vspace{2mm}

\textit{$^3$Department of Mathematics, University of California, Santa Barbara, CA 93106}\\

\vspace{2mm}

\textit{$^4$Department of Mathematics, University of Scranton, Scranton, PA 18510}\\

\vspace{2mm}

\textit{
$^5$Microsoft Station Q, Santa Barbara, CA 93106}\\

\vspace{2mm}

\textit{
$^6$Department of Mathematical Sciences, Claremont McKenna College, Claremont, CA 91711}\\

\vspace{4mm}

{\tt cuixsh@gmail.com, ktian@math.ucsb.edu, jennifer.vasquez@scranton.edu, zhenghwa@microsoft.com, hwong@cmc.edu}

\vspace{0.3cm}
}

\end{center}

\begin{abstract}
It is an open question if there are leakage-free entangling Fibonacci braiding gates.  We provide evidence to the conjecture for the negative in this paper.  We also found a much simpler protocol to generate approximately leakage-free entangling Fibonacci braiding gates than existing algorithms in the literature.
\end{abstract}

\vfill

\end{titlepage}

\section{Introduction}

Fibonacci anyons are universal for quantum computing by braidings alone \cite{freedman2002modular}.  They are conjectured to exist in fractional
quantum Hall liquids at $\nu=\frac{12}{5}$ \cite{read1999beyond},
superconductor networks \cite{mong2014universal}, and Majorana networks
\cite{hu2018fibonacci}.  Quantum algorithms such as Shor's factoing algorithm
written for the quantum circuit model are not convenient for implementation
using Fibonacci anyons because explicit qubit structure is required. 
Moreover, the universality proof of Fibonacci anyons only guarantees
efficient approximations of two-qubit entangling gates, though this is probably adequate for all practical purposes.  It has long been an interesting open
question if there are leakage-free entangling Fibonacci braiding
gates\footnote{We are not going to touch on any other variations of the question such as using measurements and/or ancillary states.}.  

In this paper, we focus on two complementary questions: proving the non-existence
of leakage-free Fibonacci entangling gates, and finding protocols to generate
good approximations adequate for the experimental construction of a Fibonacci
quantum computer.  On the first question, we found a systematic construction
of leakage-free  braiding gates, which are then proved to be non-entangling.
We also set up a computer search with up-to-date computing technology and
found no leakage-free entangling gates either. These two results provide evidence that such leakage-free
Fibonacci braiding gates do not exist. On the second question, we 
discovered  a  much  simpler  protocol  to  generate  approximately
leakage-free entangling Fibonacci braiding gates than algorithms in the
existing literature \cite{reichardt2012systematic,carnahan2016systematically}.
  The time complexity of our
approximation  algorithm for a leakage-free entangling gate is comparable to
the standard Solovay-Kitaev algorithm; however, our algorithm performs worse
for the length of words.  The gain in simplicity and geometric intuition
justifies such a sacrifice.

After recalling some basic background on Fibonacci anyons in Sec. \ref{sec2}, we search for leakage-free braiding gates in Sec. \ref{sec:3} both analytically and numerically.  In Sec. \ref{sec:approximate}, we adapt the magical iteration from \cite{reichardt2012systematic} to a more general situation in order to find approximate 2-qubit leakage-free braiding gates.  In the last section, we conjecture that our approximation algorithm should work for more general anyons such as those in $SU(2)_k$.  We  also provide a precise formulation of the tension between universality and entangling leakage-free braiding gates for anyons.

\section{Background} \label{sec2}

\subsection{Fibonacci Anyons}
There are numerous references on topological quantum computation. See, for instance, \cite{rowell2018mathematics} among others. In particular, see \cite{delaney2016local} for an explicit setup, encoding, and calculations with anyons. An anyon system, or a unitary modular tensor category, is characterized by fusion rules, $F$-matrices, $R$-matrices, topological twists, etc. 

The Fibonacci anyon system is one of the most important and also the most elegant theories for topological quantum computation \cite{freedman2002modular,trebst2008short}.  
It consists of two anyon types, $\unit$ and $\tau$, where $\unit$ represents the vacuum and $\tau$ is a non-Abelian anyon \footnote{Strictly speaking, we need to distinguish anyon types vs anyons or (quasi)-particles  \cite{wang2010topological}.  But for Fibonacci anyons, this difference can be safely ignored.}. The only nontrivial fusion rule is $\tau \otimes \tau = \unit \oplus \tau$. For anyons $a,\ b,\ c,\ d,$ $(a,b,c;d)$ is called admissible if $d$ is a total type of $a \otimes b \otimes c$; that is, $d$ is an outcome of fusing $a$, $b$, and $c$.   If $(a,b,c;d)$ is admissible, then the $F$-matrix $F^{abc}_d$ is the $1 \times 1$ identity matrix whenever $a,\ b, \ c,$ or $d$ is $\unit$, and,
\begin{equation}
F:= F^{\tau\tau\tau}_{\tau} = 
\begin{pmatrix}
 \phi^{-1} & \sqrt{\phi^{-1}} \\
 \sqrt{\phi^{-1}} & - \phi^{-1} 
\end{pmatrix},
\end{equation}
where $\phi = \frac{1 + \sqrt{5}}{2}$ is the golden ratio.
 Note that $F$ is a real symmetric and involutary matrix.
 For $R$-symbols, we have $R^{\unit a}_a = R^{a \unit}_{a} = 1,$ $\Rttone = e^{-\frac{4\pi i}{5}},$ and $\Rttt = e^{\frac{3\pi i}{5}}$.  Denote by $R = \diag(\Rttone, \Rttt)$.

\subsection{Encoding of a Qubit}
To encode one qubit, we take three $\tau$ particles with total type $\tau$. The corresponding Hilbert space $V^{\tau\tau\tau}_{\tau}$ (or $\Hom(\tau, \tau\otimes \tau\otimes \tau)$) has dimension $2$.  We will describe two bases for $V^{\tau\tau\tau}_{\tau}$ using splitting/fusion trees.  

The first (splitting/fusion tree) basis for $V^{\tau\tau\tau}_{\tau}$ is denoted by  $\CatB_L$ and can be described as follows. We first split a $\tau$ into a pair of anyons $(x, \tau)$, and then continue to split $x$ into a pair $(\tau, \tau)$.  The splitting/fusion tree for this basis is illustrated on the lefthand side of  Figure~\ref{fig:TwoBases}. One can also think of the fusion process in reverse, namely, one fuses the first two $\tau\,'$s into $x$, and then fuses $x$ and the third $\tau$ into $\tau$. According to the fusion rules, $x$ could be either $\unit$ or $\tau$. Denote by $\ket{x}_L$ the basis element corresponding to the splitting/fusion process mentioned above. Then $\CatB_L := \{\ket{\unit}_L, \ket{\tau}_L\}$ is an orthonormal basis for $V^{\tau\tau\tau}_{\tau}$. We can encode a qubit $\Complex^2$ in $V^{\tau\tau\tau}_{\tau}$ by the map, $\ket{0} \mapsto \ket{\unit}_L$, $\ket{1} \mapsto \ket{\tau}_L$. 

Similarly, there is a different basis $\CatB_R$, shown on the righthand side of Figure~\ref{fig:TwoBases},  where one splits $\tau$ into $(\tau, y)$ followed by splitting $y$ into $(\tau,\tau)$. Again, $y$ can be either $\unit$ or $\tau$.  Denote by $\ket{y}_R$ the corresponding the basis element and $\CatB_R = \{\ket{\unit}_R, \ket{\tau}_R\}$. Both $\CatB_L$ and $\CatB_R$ are called the computational bases for the one-qubit space $V^{\tau\tau\tau}_{\tau}$.  They are related by the matrix $F$:
\begin{equation}
\ket{y}_L = \sum\limits_{x=\unit,\tau} F_{xy} \ket{x}_R
\end{equation}
for $y = \unit, \tau$, and where it is understood that $F_{\unit\unit} = F_{11}$, $F_{\unit \tau} = F_{12}$,$F_{ \tau \unit} = F_{21}$, and $F_{\tau \tau} = F_{22}$. 
\begin{figure}
\centering
\includegraphics[width = 1 \linewidth]{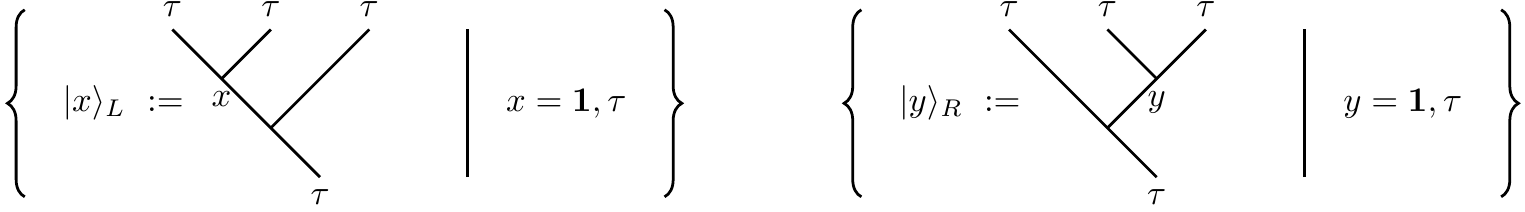}
\caption{Two splitting/fusion tree bases for $V_{\tau}^{\tau\tau\tau}$.}\label{fig:TwoBases}
\end{figure}

\begin{figure}
    \centering
    \includegraphics[width=1in]{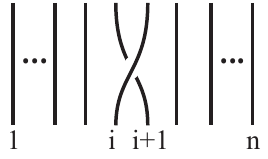}
    \caption{Braid generator $\sigma_{i, n}$}
    \label{fig:BraidGenerator}
\end{figure}

We next describe the action of the braid group.  Recall that the $n$-strand braid group $B_n$ has the presentation,
\begin{equation}
B_n = \langle \sigma_1, \cdots, \sigma_{n-1}\ | \ \sigma_{i}\sigma_{i+1}\sigma_{i} = \sigma_{i+1}\sigma_{i}\sigma_{i+1}, \ \sigma_{i}\sigma_{j} = \sigma_{j}\sigma_{i}, |i-j|>1 \rangle, 
\end{equation}
where the convention is that $\sigma_i$ corresponds to the braid diagram such that the $i$-th strand goes over the $(i+1)$-th strand, as illustrated in Figure~\ref{fig:BraidGenerator}.

The encoding of the three $\tau$ particles described above leads to a unitary representation of the three-strand braid group,
\begin{equation}
\rho_{3}\colon B_3 \longrightarrow \U(V^{\tau\tau\tau}_{\tau}). 
\end{equation}
 Denote by $\rho_3^L(\sigma)$ (resp. $\rho_3^{R}(\sigma)$) the matrix of a braid $\sigma$ under the basis $\CatB_L$ (resp. $\CatB_R$). Then,

\begin{equation}\label{eq:rho3L_sigma1}
    \rho_3^L(\sigma_1) = \rho_3^R(\sigma_2) = R = \diag(\Rttone, \Rttt),
\end{equation}
\begin{equation}\label{eq:rho3L_sigma2}
    \rho_3^L(\sigma_2) = \rho_3^R(\sigma_1) = FRF =
    \left(
    \begin{array}{cc}
    e^{\frac{4\pi i}{5}}\phi^{-1} & e^{-\frac{3\pi i}{5}}\sqrt{\phi^{-1}} \\
    e^{-\frac{3\pi i}{5}}\sqrt{\phi^{-1}}    & - \phi^{-1} \\
    \end{array}\right)
\end{equation}

Thus, under the two bases $\CatB_L, \CatB_R$, the matrices of $\sigma_1$ and $\sigma_2$ are swapped. They generate the same group under either basis, so that there is essentially no difference between $\CatB_L$ and $\CatB_R$. As a default convention, by computational basis, we will take to mean $\CatB_L$ unless explicitly stated otherwise.  The matrices $\rho_3(\sigma):=\rho_3^L(\sigma)$ are called $1$-qubit quantum gates.

It is well-known that the $\rho_3(\sigma_1)$ and $\rho_3(\sigma_2) $ generate a dense subgroup of $\U(2)$ up to phases \cite{freedman2002modular}.  Interestingly, in the $F$-matrix of the Fibonacci theory lies in the image.  Explicitly, it follows from the identities  $(RF)^3 = \Rttone I_2$ and $F^2 = I_2$ that 
\[\rho_3(\sigma_1\sigma_2\sigma_1)=\Rttone F.\]
Moreover, \cite{kliuchnikov2014asymptotically} provides an  asymptotically optimal algorithm which approximates an arbitrary unitary matrix using products of the generators $\rho_3(\sigma_1)$ and $\rho_3(\sigma_2)$ and characterizes the exact image of $B_3$ from the Fibonacci theory.

\subsection{Encoding of 2-qubits}\label{sec:2q}

Let $\Sw \in \U(\Complex^2 \otimes \Complex^2)$ be the 2-qubit gate mapping $\ket{i,j}$ to $\ket{j,i}$, $i,j = 0,1$.  Alternatively, $\Sw$ is the $4 \times 4$ permutation matrix obtained by exchanging the second and third rows of a $4 \times 4$ identity matrix.  

Recall that a 2-qubit gate $U \in \U(\Complex^2 \otimes \Complex^2)$ is called \emph{non-entangling} if one of the following conditions is satisfied (and the other condition will hold as a consequence).
\begin{enumerate}
    \item $U$ is of the form $A \otimes B$ or $\Sw \circ (A \otimes B)$ for some $1$-qubit gates $A,B \in \U(\Complex^2)$.
    \item $U$ maps product states to product states. That is, for any $\ket{x}, \ket{y} \in \Complex^2$, there exist $\ket{u}, \ket{v} \in \Complex^2$ such that $U(\ket{x} \otimes \ket{y}) = \ket{u} \otimes \ket{v}$.
\end{enumerate}
$U$ is called \emph{entangling} otherwise.  Note that the non-entangling gates form a subgroup.

All 1-qubit gates together with any entangling 2-qubit gate is universal. Hence any universal gate set for 1-qubit gates plus an entangling 2-qubit gate is a universal gate set for all qubits. This shows that entangling gates are essential for quantum computing, and in this paper, we investigate whether such entangling 2-qubit gates can arise from the Fibonacci theory. 

In particular, we are concerned with the encoding of 2-qubits obtained from six $\tau$ particles from the Fibonacci theory with total type trivial. Explicitly, we group the first three $\tau$ particles to form the first qubit and group the last three to form the second qubit. We further require the total type of each group of anyons to be trivial. The resulting Hilbert space $V_{\unit}^{\tau^{\otimes 6}}$ of six $\tau$ particles with total type trivial has dimension five. 
The four in Figure~\ref{fig:TwoQubits} are denoted by  $\ket{\unit\unit}, \ket{\unit\tau}, \ket{\tau\unit},\ket{\tau\tau}$ and span the \emph{computational subspace} $V_C$. The element $\ket{NC}$ in Figure~\ref{fig:NC} we call the \emph{non-computational state}.  Thus $V_{\unit}^{\tau^{\otimes 6}} = \text{span}\{\ket{NC}\} \oplus V_C$.

The computational subspace $V_C$ encodes 2-qubits in the way described in Figure~\ref{fig:TwoQubits}.  Note that the basis $\CatB_L$ is used for the first qubit, while $\CatB_R$ for the second qubit.  As mentioned in the previous subsection, there is essentially no difference between the two bases. The particular choice here is simply for notational convenience.  To emphasize this encoding of  two qubits, we will write $V_C= V^{\tau\tau\tau}_{\tau} \otimes V^{\tau\tau\tau}_{\tau}$.  

\begin{figure}
\centering
\includegraphics[scale = 1]{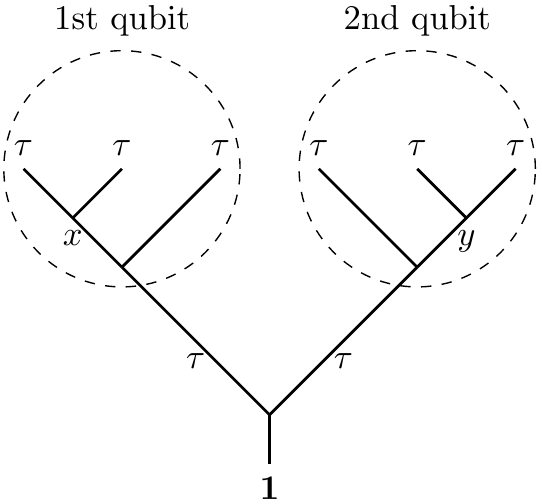}
\caption{The encoding of two qubits where $x,y = \unit, \tau$. }\label{fig:TwoQubits}
\end{figure}
\begin{figure}
\centering
\includegraphics[scale = 1]{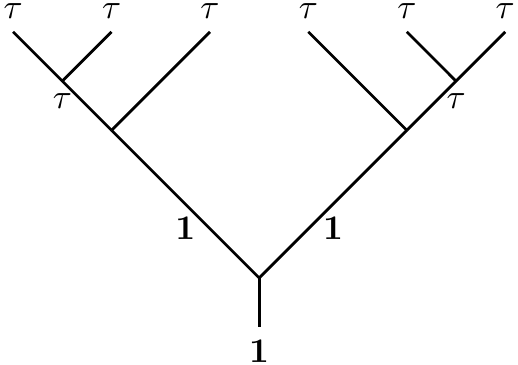}
\caption{The non-computational basis element.}\label{fig:NC}
\end{figure}

By braiding, we obtain a unitary representation of the six-strand braid group,
\begin{equation}
\rho_{6}\colon B_6 \longrightarrow \U(V^{\tau^{\otimes 6}}_{\tau}). 
\end{equation}
Let  $P_{14}$ be the permutation matrix obtained by exchanging the first and fourth rows of a $5 \times 5$ identity matrix.  Recall that $I_2$ is the $2 \times 2$ identity matrix.    By convention, the tensor product $A \otimes B $ is the matrix of the form $(a_{ij} B)$.

Direct calculation shows that the matrices of the braid group generators under the basis  $\{ \ket{NC}, \ket{\unit\unit}, \ket{\unit\tau}, \ket{\tau\unit},\ket{\tau\tau} \}$ are represented by,
\begin{eqnarray}
\rho_6(\sigma_1) &=& (\Rttt) \oplus (R \otimes I_2) \label{eqn:rhosigma1}\\
\rho_6(\sigma_2) &=& (\Rttt) \oplus (FRF \otimes I_2)\\
\rho_6(\sigma_3) &=& P_{14} \left((\Rttt) \oplus R \oplus FRF\right) P_{14}\\
\rho_6(\sigma_4) &=& (\Rttt) \oplus (I_2 \otimes FRF)\\
\rho_6(\sigma_5) &=& (\Rttt) \oplus (I_2 \otimes R).
\end{eqnarray}

Note that the formula for $\rho_6(\sigma_3)$ means that when restricting to the subspace $\text{span}\{\ket{NC}, \ket{\tau\tau}\}$ it is equal to $\rho_3(\sigma_2) = FRF$. We will use this fact later in Section \ref{sec:approximate}.


\begin{definition}
A unitary acting on $V^{\tau^{\otimes 6}}_{\tau}$ is called \emph{leakage-free} if it preserves the 4-dimensional computational subspace $V_C$.
\end{definition}

Equivalently, a unitary is leakage-free if its $(1,1)$-entry has norm equal to~$1$. To perform quantum computing, we need to have leakage-free gates to avoid information leakage. We also allow the states to go out of the computational subspace temporarily if they are performed in a controlled way. 



In the Fibonacci 2-qubit model, if a braiding gate $\rho_6(\sigma)$ is leakage-free, then we say it is entangling if the restriction of $\rho_6(\sigma)$ on $V_C$ is entangling with respect to the decomposition $V_C = V^{\tau\tau\tau}_{\tau} \otimes V^{\tau\tau\tau}_{\tau}$.  For example, we see from Equation (\ref{eqn:rhosigma1}) for the first braid generator $\sigma_1$ produces a leakage-free gate. However, it is not entangling since $\rho_6(\sigma_1)\res{V_C} = R \otimes I_2$.

It has been long suspected that, in the Fibonacci model, there are no braids that realize exactly leakage-free entangling gates. Our results in the next section support such a possibility.

\section{Leakage-free gates}\label{sec:3}

The formulas from Section~\ref{sec2} for the gates $\rho_6(\sigma_1)$, $\rho_6(\sigma_2)$, $\rho_6(\sigma_4)$, and $\rho_6(\sigma_5)$ immediately imply that they are leakage-free and non-entangling on $V_C$.   Thus, because the non-entangling gates form a closed subgroup,  any word in the braid group generators $\sigma_1, \sigma_2, \sigma_4$ and $\sigma_5$ will also be leakage-free and non-entangling.  In this section we will consider two other braids, $\Delta$ and $\Sigma$, that also produce leakage-free, non-entangling gates.

 \begin{figure}[h]
  \centering
  
  \begin{minipage}{.6in}
    \includegraphics[width=\textwidth]{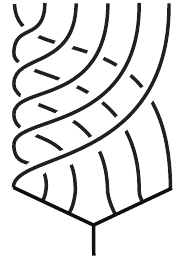}
\end{minipage}    
\; = \; 
\begin{minipage}{.6in}
    \includegraphics[width=\textwidth]{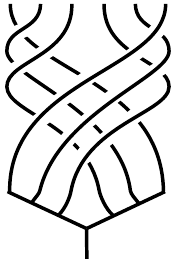}
\end{minipage}    
    \caption{The half-twist $\Delta$ applied to a splitting/fusion tree.} \label{fig:Delta}
\end{figure}

\begin{lemma} \label{lemma:Delta} Let $\Delta = \sigma_1 (\sigma_2 \sigma_1)(\sigma_3 \sigma_2 \sigma_1 )( \sigma_4\sigma_3 \sigma_2 \sigma_1)(\sigma_5\sigma_4\sigma_3 \sigma_2 \sigma_1)$.  Then \[\rho_6(\Delta) = (\Rttone)^3\cdot( I_1 \oplus\Sw) \]
\end{lemma}

\begin{proof}
 $\Delta$ is the half-twist, as illustrated on the left hand side in Figure~\ref{fig:Delta}. 
 Isotope $\Delta$ as in the ride hand side and rewrite it as 
 the product
\[\Delta = (\sigma_1 \sigma_2\sigma_1)\cdot (\sigma_5\sigma_4\sigma_5) \cdot  (\sigma_3\sigma_2\sigma_1)(\sigma_4\sigma_3\sigma_2)(\sigma_5\sigma_4\sigma_3).\] 
Recall from Section~\ref{sec2} that $\rho_6(\sigma_1 \sigma_2 \sigma_1) = {(\Rttt)^3} \oplus ( \Rttone F \otimes I_2)$, and  $\rho_6(\sigma_5\sigma_4 \sigma_5) = {(\Rttt)^3} \oplus (I_2 \otimes \Rttone F)$. Furthermore,
\[ \rho_6((\sigma_3\sigma_2\sigma_1)(\sigma_4\sigma_3\sigma_2)(\sigma_5\sigma_4\sigma_3)) = I_1 \oplus ( \Rttone (F \otimes F) \Sw) \]

With $(\Rttt)^2 = \Rttone$, the formula for $\rho_6(\Delta)$ then follows immediately. 
\end{proof}


Next, we explain the topological procedure that led us to the pure braid  $\Sigma = (\sigma_3\sigma_2\sigma_1)( \sigma_1 \sigma_2\sigma_3)$, which yields a leakage-free gate.  Start with a braid on four strands which returns the first strand to its leftmost position.  Such a braid belongs in the annular braid group, which is generated by $\sigma_1^2$, $\sigma_2$, and $\sigma_3$ in $B_4$ \cite{birman2016braids}. Now replace the first strand by three parallel strands to obtain a braid on six strands, which is a product of $\Sigma$, $\sigma_4$, and $\sigma_5$ in $B_6$. Any braid obtained in this way preserves $V_C$.  $\Sigma$ is illustrated in Figure~\ref{fig:Sigma}, and a computation yields the following lemma, from which it is also easy to see that $\Sigma$ produces a non-entangling gate.    
\begin{figure}[h]
\centering
  \begin{minipage}{.7in}
    \includegraphics[width=\textwidth]{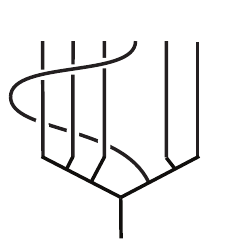}
\end{minipage}    
    \caption{The pure braid $\Sigma$}
    \label{fig:Sigma}
\end{figure}

\begin{lemma} Let $\Sigma = (\sigma_3\sigma_2\sigma_1)( \sigma_1 \sigma_2\sigma_3)$.  Then $\rho_6(\Sigma) = I_1 \oplus ( I_2 \otimes R^2)$ 
\end{lemma}

We remark that we could instead have arrived at the pure braid $\Sigma$ by starting with a braid on four strands which moves the first strand to the rightmost position, and then replacing the first strand with three parallel strands.  In that case, we produce a braid on six strands that is a product of $\Sigma$, $\sigma_4$, $\sigma_5$, and $(\sigma_3\sigma_2\sigma_1)(\sigma_4\sigma_3\sigma_2)(\sigma_5\sigma_4\sigma_3)$ in $B_6$. Recall from our proof of Lemma~\ref{lemma:Delta} that $(\sigma_3\sigma_2\sigma_1)(\sigma_4\sigma_3\sigma_2)(\sigma_5\sigma_4\sigma_3)$ can be written as a product of $\Delta$, $\sigma_1$, $\sigma_2$, $\sigma_4$, and $\sigma_5$.  Thus, while the resulting braid will also yield a leakage-free gate, it is one that we've seen already.  
\\

We summarize the above results in the following theorem.  

\begin{theorem} \label{thm:DeltaSigma}  

Any word $w$ in $\sigma_1, \sigma_2, \sigma_4, \sigma_5,\Delta$, and $\Sigma $ produces a gate that is leakage-free and  non-entangling on the computational subspace $V_C$. \\
\end{theorem}

\begin{remark}
Topological constructions similar to used in Theorem~\ref{thm:DeltaSigma} may be used to obtain braids which preserve subspaces other than $V_C$.  Often, the braids turn out to be entangling on the complement of the preserved subspace.  

In particular, to find an infinite family of braids which fixes subspace spanned by $\ket{\unit\unit}$, we may start with a pure braid on three strands and double every strand.  We may further take products with $\sigma_1, \sigma_2, \sigma_4, \sigma_5$, and $\Delta$, and still obtain gates which fix $\ket{\unit \unit}$ up to a phase.  Interestingly, unlike the situation with the non-computational $\ket{NC}$, many of the gates that fix $\ket{\unit \unit}$ up to a phase are entangling on the complementary 4-dimensional subspace.  For example, it can be shown that $\rho_6((\sigma_2 \sigma_3)^3)$ fixes $\ket{\unit \unit}$ up to a phase, does not fix $\ket{NC}$, and is entangling on the basis elements $\ket{NC}, \ket{\unit \tau}, \ket{\tau \unit}$ and $\ket{\tau \tau}$.  

To obtain braids that fix $\ket{\unit \tau}$ and $\ket{\tau \unit}$, choose a annular braid on five strands and double the first or last.  As above, many of the resulting gates are entangling on the complementary 4-dimensional subspace.  For example, $\rho_6((\sigma_2 \sigma_3)^3)$ fixes $\ket{\tau \unit}$ up to a phase,  does not fix $\ket{NC}$,  and is entangling on the basis elements $\ket{NC}, \ket{\unit \unit}, \ket{\unit \tau}$ and $\ket{\tau \tau}$. 

Although it is easy to find braids that fix $\ket{\unit \unit}$, $\ket{\unit \tau}$ and $\ket{\tau \unit}$, we do not know of any gate which fixes $\ket{\tau \tau}$ up to a phase, except for $\rho_6(\Delta)$. 

\end{remark}




\subsection{Systematic computer search}

To help find leakage-free entangling gates, we performed a computer search by enumerating elements of the braid group and computing their corresponding matrices in the representation given in Section \ref{sec2}. Then we checked whether it was leakage-free, and whether it was entangling. 

We enumerated the elements of the braid group $B_6$ by taking words consisting of the generators and their inverses. We excluded trivial cases of a generator appearing adjacent to its inverse.

Due to exponential growth rate of the number of words of a given length, our search only reached words of length 7, and no leakage-free entangling gates were found.

\section{Approximate Leakage-free Entangling Braiding Gates}
\label{sec:approximate}
In this section, we provide a simple procedure which approximates certain leakage-free entangling gates with braidings to arbitrary precision. 

\subsection{Braiding gates preserving $\text{span}\{\ket{NC}, \ket{\tau \tau}\}$}
\label{subsec:braiding V}
For the $6$-anyon encoding of two qubits as shown in Figures \ref{fig:TwoQubits} and \ref{fig:NC}, we consider braiding gates that preserve the subspace $V := \text{span}\{\ket{NC}, \ket{\tau \tau}\}$. Let $V^{\perp} = \text{span}\{ \ket{\unit \unit},  \ket{\unit \tau}, \ket{\tau \unit} \}$.

\begin{figure}
    \centering

\begin{minipage}{.75in}
    \includegraphics[width=\textwidth]{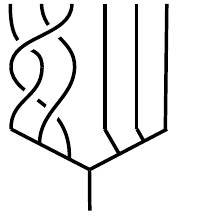}
\end{minipage}    
\; = \; 
\begin{minipage}{.75in}
    \includegraphics[width=\textwidth]{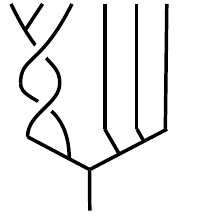}
\end{minipage}    

    \caption{The braid $\sigma_2\sigma_1\sigma_1\sigma_2$ applied to a splitting/fusion tree}
    \label{fig:2112}
\end{figure}

First, consider the braid $\sigma_2\sigma_1\sigma_1\sigma_2$, which is represented as in Figure~\ref{fig:2112} where the equality is obtained by isotopy of braids. Then direct computation shows that with respect to the decomposition $V \oplus V^{\perp}$,
\begin{equation}
\rho_6(\sigma_2\sigma_1\sigma_1\sigma_2) = \rho_3(\sigma_1^2) \oplus \diag(1,1, (R^{\tau\tau}_{\tau})^2).
\end{equation}
Similarly, 
\begin{equation}
\rho_6(\sigma_4\sigma_5\sigma_5\sigma_4) = \rho_3(\sigma_1^2) \oplus \diag(1, (R^{\tau\tau}_{\tau})^2, 1).
\end{equation}
It can also be verified that $\rho_6(\sigma_3)$ preserves the decomposition $V \oplus V^{\perp}$, where
\begin{equation}
\rho_6(\sigma_3) = \rho_3(\sigma_2) \oplus \diag(\Rttone, \Rttt, \Rttt).
\end{equation}
Hence, through braidings from the 6-anyon encoding of two qubits, we can obtain all of the group of gates generated by $\{ \rho_3(\sigma_1^2),  \rho_3(\sigma_2) \}$ on $V$. We do not know if this group contains all the possible braiding gates on $V$. However, Proposition \ref{prop:universality} below implies that $\{\rho_3(\sigma_1^2), \rho_3(\sigma_2)\}$ is already a universal gate set on $V$.

In particular, recall the well-known result that $\{\rho_3(\sigma_1), \rho_3(\sigma_2)\}$ generates a dense subgroup of $\SU(2)$ up to phases \cite{freedman2002modular}. We prove a stronger result in the following proposition. 
\begin{proposition}
\label{prop:universality}
Let $\rho_3(\sigma_1) = \rho_3^L(\sigma_1), \rho_3(\sigma_2) = \rho_3^L(\sigma_2)$ be the $1$-qubit gates given in Equations \ref{eq:rho3L_sigma1}, \ref{eq:rho3L_sigma2}. Then $\{\rho_3(\sigma_1^2), \rho_3(\sigma_2^2)\}$ generate a dense subgroup of $\SU(2)$ up to global phases.
\begin{proof}
Let $U_1, U_2 \in \SU(2)$. By the classification of subgroups of $\SU(2)$, if $U_1$ and $U_2$ have infinite order and they do not commute up to phases, then $\{U_1, U_2\}$ generate a dense subgroup of $\SU(2)$. Take $U_1 = \rho_3(\sigma_1^2 \sigma_2^4), U_2 = \rho_3(\sigma_1^2 \sigma_2^6)$. Then it is straightforward to check $U_1$ and $U_2$ do not commute.

To show that they have infinite order, we show that their eigenvalues are not $m$-th roots of unity for any integer $m$, or equivalently that their real parts are not the cosine of a rational multiple of $\pi$. Normalizing determinants to equal 1, the real part of the eigenvalues of $e^{\frac{i \pi}{10}} \rho_3(\sigma_1^2 \sigma_2^4)$ and $e^{\frac{i \pi}{10}} \rho_3(\sigma_1^2 \sigma_2^6)$ are given (respectively) by:

\[ \frac{-2 + \sqrt{5}}{2} \text{ and } \frac{-3 + \sqrt{5}}{2} \]

Neither real part given above is the value of cosine at a rational multiple of $\pi$, and thus it follows from  \cite{tangsupphathawat2014algebraic} that both elements are infinite order.
\end{proof}
\end{proposition}

In Section \ref{subsec:approximation}, we will combine the fact that $\{\rho_3(\sigma_1^2), \rho_3(\sigma_2)\}$ is a universal gate set on $V$  together with some techniques developed in Section \ref{subsec:iteration} to provide a simple scheme to approximate certain 2-qubit leakage-free, entangling gates using braidings.

\subsection{Iteration to diagonal gates}
\label{subsec:iteration}
Let $D \in \U(2)$ be any diagonal gate and write it as $D = \gamma\diag(e^{-i \frac{\theta}{2}}, e^{i \frac{\theta}{2}})$ for $ -\pi \leq \theta \leq \pi$ and  $\gamma \in \U(1)$. The phase $\gamma$ will not play a role below, so we also write $D = D(\theta)$. Let $U_0 \in \U(2)$ be any 1-qubit gate.  Consider the sequence $\{U_k\}_{k=0}^{\infty}$ defined inductively by the formula:
\begin{align}
\label{equ:U_k}
    U_{k+1} = U_{k}\cdot D(\theta) \cdot U_{k}^{-1}\cdot  D(\theta)\cdot  U_{k} \cdot D(\theta)^{-2}
\end{align}
Obviously, $U_{k}$ does not depend on the phase $\gamma$. For $\theta = 0$, then $U_{k} = U_0$ for all $k$.
\begin{lemma}
\label{lem:iteration}
If $-\frac{\pi}{2} < \theta < \frac{\pi}{2}$, $\theta \neq 0$, and $|(U_0)_{12}| < 1$, then the sequence $\{U_k\}$ defined in Equation~(\ref{equ:U_k}) converges to a diagonal gate.
\begin{proof}
It suffices to consider the case $U_0 \in \SU(2)$ since by Equation~\ref{equ:U_k}, if $U_k$ has a global phase, then $U_{k+1}$ has the same global phase.

Let $\lambda = e^{i \theta}$, $\delta = |(U_0)_{12}| < 1$, and 
\begin{equation}
    U_{k} = \left(
    \begin{array}{cc}
    a_k & -\overline{b_k} \\
    b_k & \overline{a_k}\\
    \end{array} \right).
\end{equation}
We first show that there exists $\epsilon = \epsilon(\theta, \delta) < 1$ such that $|b_{k+1}| \leq \epsilon |b_{k}|$, which implies that $\{|b_{k}|\}$ converges to $0$. By direct calculation, 
\begin{align*}
    |b_{k+1}| = |b_k| y_k,
\end{align*}
where
\begin{align}
\label{equ:y_k}
    y_k &= |(1-|b_k|^2)(1-\lambda + \lambda^2) + |b_k|^2 \lambda| \\
        &= |(\lambda + \bar{\lambda} - 2)(1-|b_k|^2)+1| \\
        &= |(2-2\cos(\theta))(1-|b_k|^2) - 1|.
\end{align}
It is clear that $y_k \leq 1$. Hence $|b_{k+1}| \leq |b_{k}| \leq  \delta$. In turn, setting  $\epsilon := \max\{|1-2\cos(\theta)|, (2-2\cos(\theta))(1-\delta^2) - 1\}$, we have $y_k \leq \epsilon$. By our assumption on $\theta$, both of the two expressions in $\max\{\cdot,\cdot\}$ are strictly less than one, and hence $\epsilon < 1$.

That $|b_{k+1}| \leq \epsilon|b_k|$ implies the statement in the lemma. Intuitively, when $k$ gets large, $U_k$ is close to a diagonal gate, and hence approximately commutes with $D(\theta)$. By Equation \ref{equ:U_k}, $U_{k+1}$ would be approximately equal to $U_{k}$. The following is a more elementary argument. Again by direct calculations, 
\begin{align}
    a_{k+1} &= a_k(1- |b_k|^2(\lambda-1)^2).
\end{align}
Hence,
\begin{align}
    |a_{k+1}-a_{k}| \ =\ |a_k|\cdot |\lambda-1|^2 \cdot |b_k|^2  \ \leq \ c \epsilon^{2k} 
\end{align}
for some constant $c > 0$, which implies that the sequence $\{a_k\}$ converges.  
\end{proof}
\end{lemma}
A few remarks are in order.
\begin{remark}
For $\theta = \frac{\pi}{3}$, by Equation \ref{equ:y_k}, we have $y_k = |b_k|^2$ and hence $|b_{k+1}| = |b_k|^3$.  In this case, the sequence $\{b_k\}$ converges to $0$ exponentially faster than it does for a general $\theta$ as in the proof of Lemma \ref{lem:iteration}.  The formula in Equation \ref{equ:U_k} for $\theta = \frac{\pi}{3}$ was used in \cite{reichardt2005quantum} as a scheme to approximate certain diagonal gates. To be precise, the formula in \cite{reichardt2005quantum} does not have the \lq $D^{-2}$ ' factor as in Equation \ref{equ:U_k}. This does not change the fact that the off-diagonal entries of $U_k$ converges to zero. However, without the \lq $D^{-2}$ ' factor, the $\{U_k\}$ sequence does not converge to a diagonal gate, but rather fluctuates among several diagonal gates which differ by some powers of $D$ from each other.
\end{remark}
\begin{remark}
In \cite{reichardt2012systematic, carnahan2016systematically}, a formula different from that in Equation \ref{equ:U_k} was provided to give rise to a sequence $\{U_k\}$ which converges at an even higher rate: $|(U_{k+1})_{1,2}| = |(U_{k})_{1,2}|^5$ for $\theta = \frac{\pi}{5}$. However, their formula does not apply here. This is because $D(\frac{\pi}{5}) = \rho(\sigma_1)^3$ up to phases, and as will be seen in Section \ref{subsec:approximation}, we will give a scheme to approximate 2-qubit entangling gates with braids that preserve the subspace $V := \text{span}\{\ket{NC}, \ket{\tau\tau}\}$. However, the braids that preserve the subspace $V$ do not seem to realize the gate $\rho(\sigma_1)^3$ on $V$, but only $\rho(\sigma_1)^2$ instead.
\end{remark}
\begin{remark}
There is a geometric interpretation of the formula in Equation \ref{equ:U_k}. If we think of a 1-qubit gate $U \in \SU(2)$ as a rotation in $\Real^3$, then $D(\theta)$ is a rotation around the $z$-axis by the angle $\theta$. A unitary $U$ has an axis in the $xy$-plane if and only if its $(1,2)$-entry has norm one. Then by Lemma \ref{lem:iteration}, as long as $\theta$ has absolute value strictly between 0 and $\frac{\pi}{2}$ and the axis of $U_0$ is not in the $xy$- plane, then each iteration in Equation~(\ref{equ:U_k}) brings the axis of $U_k$ closer to the $z$-axis. In the limit $U_k$ becomes a rotation around the $z$-axis.
\end{remark}

\subsection{Approximation of 2-qubit  Leakage-free Entangling Braiding Gates}
\label{subsec:approximation}

We provide a scheme to approximate certain 2-qubit  leakage-free entangling gates with braidings. Of course, since the Fibonacci model is universal, one can in principle approximate arbitrary $n$-qubit gates using (for instance) the Solovay-Kitaev algorithm. However, the procedure we give is more explicit and simpler.  

We use the braiding gates from $\mathcal{G} :=\langle \rho_6(\sigma_2\sigma_1\sigma_1\sigma_2), \rho_6(\sigma_3) \rangle$ for the approximation. Recall that $V = \text{span}\{\ket{NC}, \ket{\tau\tau}\}$, $V^{\perp} = \text{span}\{ \ket{\unit \unit},  \ket{\unit \tau}, \ket{\tau \unit} \}$, and that gates in $\mathcal{G}$ all preserve $V$.  Choose any gate $\tilde{U}_0$ and a diagonal gate $\tilde{D}$ in $\mathcal{G}$ such that $D:=\tilde{D}\res{V}$ and $U_0:= \tilde{U}_0\res{V}$ satisfy the conditions in Lemma \ref{lem:iteration}. We then obtain a sequence of gates $\{\tilde{U}_k\}$ by the formula in Equation \ref{equ:U_k} starting from $\tilde{U}_0$ and $\tilde{D}$. Note that $\tilde{U}_k = \tilde{U}_0$ on $V^{\perp}$ for all $k\,'$s. By Lemma \ref{lem:iteration}, $\{\tilde{U}_k\}$ converges to some $\tilde{U}$ such that $\tilde{U}\res{V}$ is a diagonal gate and $\tilde{U}\res{V^{\perp}} = \tilde{U}_0\res{V^{\perp}}$ is also a diagonal gate. Hence $\tilde{U}$ is a leakage-free diagonal gate. In general it is straightforward to check whether $\tilde{U}$ is entangling for each particular choice of $\tilde{D}$ and  $\tilde{U}_0$ since $\tilde{U}$ agrees with $\tilde{U}_0$ on $V^{\perp}$. If $\tilde{U} = \diag(\lambda_{-1},\lambda_0, \lambda_1,\lambda_2,\lambda_3)$ under the basis $\{\ket{NC}, \ket{\unit\unit},\ket{\unit\tau}, \ket{\tau\unit}, \ket{\tau\tau}\}$, then $\tilde{U}$ is entangling if and only if $\lambda_3 \neq \lambda_1\lambda_2\lambda_0^{-1}$.

\begin{theorem}
Let $\tilde{D} = \rho_6(\sigma_2\sigma_1\sigma_1\sigma_2)^3, \tilde{U}_0 = \rho_6(\sigma_3)$. Then the limit of the sequence $\{\tilde{U}_{k}\}$ defined by Equation~(\ref{equ:U_k}) exists and its limit $\tilde{U}$ is a leakage-free entangling 2-qubit gate. 
\begin{proof}
With respect to the decomposition $V \oplus V^{\perp}$, we have
\begin{equation}
\tilde{D} = e^{-\frac{3\pi i}{5}}
\begin{pmatrix}
e^{-\frac{\pi i}{5}} & 0 \\
0 & e^{\frac{\pi i}{5}}\\
\end{pmatrix}
\oplus
\begin{pmatrix}
1 & 0 & 0\\
0 & 1 & 0\\
0 & 0 & e^{-\frac{2\pi i}{5}}\\
\end{pmatrix}
\end{equation}

\begin{equation}
\tilde{U}_0 = 
\left(
    \begin{array}{cc}
    e^{\frac{4\pi i}{5}}\phi^{-1} & e^{-\frac{3\pi i}{5}}\sqrt{\phi^{-1}} \\
    e^{-\frac{3\pi i}{5}}\sqrt{\phi^{-1}}    & - \phi^{-1} \\
    \end{array}\right)
\oplus
\begin{pmatrix}
e^{-\frac{4\pi i}{5}} & 0 & 0\\
0 & e^{\frac{3\pi i}{5}} & 0 \\
0 & 0 & e^{\frac{3\pi i}{5}} \\
\end{pmatrix}
\end{equation}

The angle of $\tilde{D}\res{V}$ is $\theta = \frac{2 \pi}{5}< \frac{\pi}{2}$, and the $(1,2)$-entry of $\tilde{U}_0\res{V}$ (that is, the $(1,5)$-entry of $\tilde{U}_0$) has absolute value $\sqrt{\phi^{-1}} \,\approx \, 0.786 < 1$. Hence the conditions in Lemma \ref{lem:iteration} are satisfied. $\tilde{U}$ is entangling if and only if $\tilde{U}_{5,5} \neq 1$, which can be checked numerically. 
\end{proof}
\end{theorem}

\section{$\SU(2)_k$ anyons}
As a modular tensor category, the Fibonacci theory $\text{Fib}$ is a sub category of the anyon theory $\SU(2)_3$ whose anyon types are given by $\{0,1,2,3\}$. Explicitly, the correspondence is $\unit \leftrightarrow 0,\ \tau \leftrightarrow 2$. Moreover, $\{0,3\}$ forms the semion theory $\mathcal{S}$ and $\SU(2)_3 = \text{Fib} \boxtimes \mathcal{S}$. Also note that semion $\mathcal{S}$ is an Abelian theory and $1 = 2 \otimes 3 = 2 \boxtimes 3$. Then an important observation is as follows. In the encoding of one- and two-qubit models (Section \ref{sec:2q}), if we replace all the anyons of type $\tau$ (i.e., type 2 ) by anyons of type $1$, then the braiding gates remain the same up to (irrelevant) global phases which are contributed by the semion theory. This means that for anyons of type $1$, all the results discussed in the paper still hold. 

Now for the sequence of anyon theories $\SU(2)_k$, for $k \geq 2$ with anyon types $\{0,1,\cdots, k\}$, exactly the same models of one and two qubits (and more generally $n$-qubits) as in Section \ref{sec2} can be defined with type 1 anyons. It is known that the type $1$ anyon in $\SU(2)_k$ is braiding universal if and only if $k = 3$ or $k \geq 5$ \cite{freedman2002two}. We believe that the results presented in this paper still hold for $k \geq 5$. For instance,  $\{\rho_3(\sigma_1^2), \rho_3(\sigma_2^2)\}$ generates a dense subgroup of $\SU(2)$. Also, the method for approximating entangling leakage-free 2-qubit gates in earlier sections also applies.

\subsection{A Conjecture}
Let $\CatC$ be an anyon theory, namely, a unitary modular tensor category, and $a,b,c \in \CatC$ be anyon types. Assume $c$ is a total type of $(b,b)$. Consider the embedding $V^{a^{\otimes n}}_b \otimes V^{a^{\otimes n}}_b \subset V^{a^{\otimes 2n}}_c$ for some $n > 1$. See Figure \ref{2qudits}. We treat each $V^{a^{\otimes n}}_b$ as a qudit space. We call an anyon type $a$ to have the property of {\it \Ent} if for some $n> 1$ and anyon types $b,c$, there exists a braid $\sigma \in B_{2n}$ such that the representation of $\sigma$ on $V^{a^{\otimes 2n}}_c$ preserves, and is entangling on, the subspace $V^{a^{\otimes n}}_b \otimes V^{a^{\otimes n}}_b$.

\begin{figure}
\centering
\begin{tikzpicture}[scale = 1, thick]
\draw (0,0)node[below]{$c$} -- (0,1);
\draw (0,1) -- (1.5, 2.5) node[pos = 0.5, below]{$b$};
\draw (0,1) -- (-1.5, 2.5) node[pos = 0.5, below]{$b$};
\draw (0.5, 2.5) rectangle (2.5, 3);
\draw (-2.5, 2.5) rectangle (-0.5, 3);

\draw (0.8, 3) -- (0.8, 3.4) node[above]{$a$};
\draw (2.2, 3) -- (2.2, 3.4)node[above]{$a$};
\draw (1.5,3.2) node{$\cdots$};

\draw (-0.8, 3) -- (-0.8, 3.4) node[above]{$a$};
\draw (-2.2, 3) -- (-2.2, 3.4)node[above]{$a$};
\draw (-1.5,3.2) node{$\cdots$};
\end{tikzpicture}
\caption{Two qudits}\label{2qudits}
\end{figure}
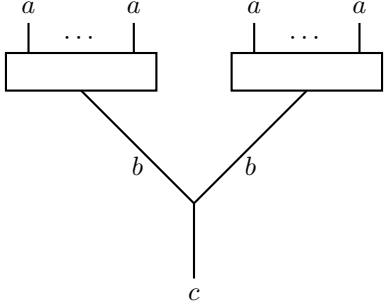

By the results in this paper, we believe that the Fibonacci anyon (or the type 1 anyon\footnote{spin $1/2$ in physics parlance.} in $\SU(2)_3$) does not have the property of \Ent. On the other hand, the type 1 anyon in $\SU(2)_k$ does have the property of \Ent \;for $k=2$ and $k=4$ \cite{wang2010topological,cui2015universal}. Thus there seems to be a tension between braiding universality and the property of \Ent. 

\begin{conjecture}
An anyon type has the property of \Ent \;if and only if the braid group representations of $B_n$ associated with it have finite images for all $n\geq 1$.
\end{conjecture}

The anyon of type $1$ of $\SU(2)_8$ has finite images for $B_3$ and $B_4$, but infinite images for all $B_n, n\geq 5$ \cite{freedman2002two}.

By the Property ${\bf F}$ conjecture \cite{naidu2011finiteness}, we can also formulate the above as:
\begin{conjecture}
An anyon type has the property of \Ent \;if and only if its quantum dimension is the square root of an integer.
\end{conjecture}

\noindent {\bf{Acknowledgement:}} K.T. and Z.W. are partially supported by NSF grant  FRG-1664351, and H.W. by DMS-1841221. S.C. acknowledges the support from the Simons Foundation.  All authors thank the IAS School of Mathematics for their support during Summer 2018, where the project started and most of the results were obtained.

\bibliographystyle{plain}
\bibliography{Fib_bib}
\end{document}